\documentclass [journal,onecolumn,11pt]{IEEEtran}
\usepackage{amsfonts,amsmath,amssymb}
\usepackage{indentfirst, setspace}
\usepackage{url,float}
\usepackage{color}
\usepackage[a4paper,ignoreall]{geometry}
\usepackage{longtable}
\usepackage{graphicx}
\usepackage[a4paper,ignoreall]{geometry}
\usepackage{multicol}
\usepackage{stfloats}
\usepackage{enumerate}
\usepackage{cite}
\usepackage{amsthm}
\usepackage{booktabs}
\usepackage{mathrsfs}
\usepackage{multirow}
\usepackage{amssymb}
\usepackage{verbatim}
\usepackage{cases}
\usepackage{xcolor}
\usepackage[square, comma, sort&compress, numbers]{natbib}
\usepackage[overload]{empheq}
\def\qu#1 {\fbox {\footnote {\ }}\ \footnotetext { From Qu: {\color{red}#1}}}
\def\hqu#1 {}
\def\kq#1 {\fbox {\footnote {\ }}\ \footnotetext { From KangQuan: {\color{blue}#1}}}
\def\hkq#1 {}

\newcommand{\mkq}[1]{{{\color{blue}#1}}}

\newtheorem{Th}{Theorem}[section]
\newtheorem{Cor}[Th]{Corollary}
\newtheorem{Prop}[Th]{Proposition}
\newtheorem{Prob}{Problem}
\newtheorem{Lemma}[Th]{Lemma}
\newtheorem{Def}{Definition}
\newtheorem{example}{Example}

\newtheorem{Rem}{Remark}

\newcommand{\tr}{{\rm Tr}}

\newcommand{\gf}{{\mathbb F}}

\makeatletter
\newcommand{\figcaption}{\def\@captype{figure}\caption}
\newcommand{\tabcaption}{\def\@captype{table}\caption}
\makeatother

\begin{document}
	\title{
		On the Differential Linear Connectivity Table \\ of Vectorial Boolean Functions}
	\author{{Kangquan Li, Chunlei Li, Chao Li and Longjiang Qu}
	\thanks{Kangquan Li, Longjiang Qu and Chao Li are with the College of Liberal Arts and Sciences,
		National University of Defense Technology, Changsha, 410073, China.
		Longjiang Qu is also with the State Key Laboratory of Cryptology, Beijing, 100878, China. Chunlei Li is with the Department of Informatics, University of Bergen, Bergen N-5020, Norway.
		E-mail: likangquan11@nudt.edu.cn, chunlei.li@uib.no, lichao\_nudt@sina.com, ljqu\_happy@hotmail.com.
	}}

	\maketitle
	
\begin{abstract}

Vectorial Boolean functions are crucial building blocks in symmetric ciphers. 
Different known attacks on block ciphers have resulted in diverse cryptographic criteria of vectorial Boolean functions,
such as differential distribution table and nonlinearity.  
Very recently, Bar-On et al. introduced at Eurocrypt'19 \cite{DLCT2019}  a new tool, called 
the Differential-Linear Connectivity Table (DLCT), which allows for taking
into account the dependency between the two subciphers $E_0$ and $E_1$,
and leads to significant improvements of differential-linear cryptanalysis attacks on ciphers ICEPOLE and 8-round DES.
This paper is a follow-up work of \cite{DLCT2019}, which presents further theoretical characterization of the DLCT of vectorial
Boolean functions and also investigates this new criterion of functions with certain forms.

{ In this paper we introduce a generalized concept of the additive autocorrelation, which is extended from Boolean functions to the vectorial Boolean functions, and 
use it as a main tool to investigate the DLCT property of vectorial Boolean functions.}
Firstly, by establishing a connection between the DLCT and the additive autocorrelation, 
we characterize properties of DLCT by means of the Walsh transform and the differential distribution table, and present generic lower bounds on the differential-linear uniformity (DLU) of 
vectorial Boolean functions. Furthermore, we investigate the DLCT property of monomials, APN, plateaued and AB functions. 
Our study reveals that the DLCT of these special functions are closely related to other cryptographic criteria. 
Next, we prove  that the DLU of vectorial Boolean functions is invariant under
the extended-affine (EA) equivalence but not invariant under the Carlet-Charpin-Zinoviev (CCZ) equivalence, and that the DLCT spectrum is only invariant under affine equivalence. In addition, under affine equivalence, we exhaust the DLCT spectra and DLU of optimal S-boxes with $4$ bit by Magma. Finally, we investigate the DLCT spectra and DLU of some polynomials over $\gf_{2^n}$, including the inverse, Gold, Bracken-Leander power functions and all quadratic polynomials.


\end{abstract}	


\section{Introduction}
{
  Let $p$ be a prime and $n, m$ two arbitrary positive integers. We denote by $\gf_{p^n}$ the finite field with $p^n$ elements and by $\gf_p^n$ the $n$-dimensional vector space over $\gf_p$. In this paper, we always identify the vector space $\gf_p^n$ with $\gf_{p^n}$. For any set $E$, we denote the nonzero elements of $E$ by $E^{*}$ and the cardinality of $E$ by $\#E$.  
  
  Vectorial Boolean functions, also  called  $(n,m)$-functions from $\gf_2^n$ to $\gf_{2}^m$, play a crucial role in block ciphers.  
Many attacks have been proposed against the diverse
block ciphers, and have led to criteria, such as low differential uniformity, high nonlinearity, high algebraic degree, etc, that the implemented cryptographic functions must satisfy.  
In Eurocrypt'18, Cid et al. \cite{BCT2018} introduced a new concept {on the cryptographic property of} S-boxes: the boomerang connectivity table (BCT) that
permits to simplify the complexity analysis on the dependency between the upper part and lower part of a block chipher in the boomerang attack.
The work of \cite{BCT2018} shortly attracted research attentions in the study of BCT property of cryptographic functions \cite{BC2018,LQSL2019,MTX2019,SQH2019}
and stimulated research progress in other cryptanalysis methods.
  Very recently, in Eurocrypt'19, Bar-On et al. \cite{DLCT2019} introduced  a new tool called the differential-linear connectivity table (DLCT) that
takes the dependency into account and to use it for making the classical differential-linear  attacks \cite{LH1994} more efficient.  
The authors also presented the relation between the DLCT and the differential distribution table (DDT) of S-boxes, indicating that each row of the DLCT  is equal (up to normalization) to the Fourier transform of the Boolean function represented by the corresponding row of the DDT. 

This paper aims to provide further theoretical characterization of the DLCT property, explicitly the DLCT spectrum and differential-linear uniformity (DLU), of vectorial Boolean functions. 
To this end, we firstly establish a connection between DLCT and a generalized concept of the \textit{additive autocorrelation}, which is extended from  Boolean functions to vectorial Boolean functions and is used as a main tool in this paper. 
Based on the study of the additive autocorrelation  of vectorial Boolean functions, we give some characterizations of the DLCT by means of the Walsh transform and the DDT, and generic lower bounds on the DLU of vectorial Boolean functions.  
Moreover, for certain functions like monomials, APN, plateaued and AB functions, we presented the relation of 
their DLCT  with other cryptographic criteria:
we show that the DLCT of monomials $x^d$ on $\gf_{2^n}$ with $\gcd\left(d,2^n-1\right)=1$ is identical to the additive autocorrelation  of $\tr_{2^n}\left(x^d\right)$, where  for any $\alpha\in\gf_{2^n}, $ $\tr_{2^n}(\alpha) = \alpha+\alpha^2+\alpha^{4}+\cdots+ \alpha^{2^{n-1}}$, and the DLU of $x^d$ becomes the absolute indicator of $\tr_{2^n}\left(x^d\right)$. The DLCT of APN and AB/plateaued functions are converted to the Walsh transform of two classes of balanced Boolean functions. 
 Next, we investigate the DLCT property of vectorial Boolean functions under affine, extended-affine (EA) and Carlet-Charpin-Zinoviev (CCZ) equivalence, and 
show that the DLCT spectrum is affine-invariant and the DLU is EA-invariant but not CCZ-invariant.
 Furthermore, based on the classification of optimal S-boxes with $4$ bit by Leander and Poschmann \cite{LP2007}, 
 we calculate their differential-linear uniformities and DLCT spectra (see Table \ref{4Autocorrelation}) by Magma, which indicates that the DLU of optimal $4$-bit S-boxes only takes value $4$ and $8$. 
 Finally, we investigate the DLCT spectra of some special polynomials over finite fields with characteristic two. It is shown in this part that the DLCT spectrum of the inverse function can be properly characterized by the Kloosterman sums, 
the DLCT spectrum of all quadratic polynomials can be completely determined and the DLCT of the Kasami APN monomials for some parameters appear to be optimal.

The rest of this paper is organized as follows. Section \ref{Preliminaries} recalls basic definitions, particularly the generalized tool of the additive autocorrelation and the  concept of DLCT, as well as some known results that we will use in our subsequent discussions. Section \ref{Characterizations} is devoted to the characterization of the DLCT:
we firstly characterize the DLCT by means of the Walsh transform and DDT and generic lower bounds on the DLU of any vectorial Boolean functions; and then study some properties about the DLCT of monomials, APN, plateaued and AB functions. 
Besides, we consider the property of invariant of the DLU, the DLCT spectrum under the affine, EA and CCZ equivalences. By Magma, we also present all possible values of the DLUs and DLCT spectra  of optimal $4$-bit S-boxes. 
In Section \ref{polynomials}, we compute  the DLCT spectra of some special polynomials. Finally, Section \ref{conclusion} draws a conclusion of our work in this paper. }
 
\section{Preliminaries}
\label{Preliminaries}

{In this section, we firstly recall some common concepts about (vectorial) Boolean functions and known results that are useful for our subsequent discussions. 
Since the vector spaces $\gf_2^n, \gf_2^m$ can be deemed as the finite fields  $\gf_{2^n}, \gf_{2^m}$ under certain bases, 
we will use the notation $\gf_2^n$ (resp. $\gf_2^m$) and $\gf_{2^n}$ (resp. $\gf_{2^m}$) interchangeably when there is no ambiguity.
We will also use the inner product $a\cdot b$, where $a,b\in\gf_2^n$, and $\tr_{2^n}(ab)$ in the context of vector spaces and finite fields  interchangeably.
}

\subsection{Walsh transform, Bent functions, AB functions and Plateaued functions}
  For any $n$-variable Boolean function $f$, its \textit{Walsh transform} from $\gf_{2}^n\to \mathbb{C}$ is defined as 
$$ W_f(\omega) = \sum_{x\in\gf_{2}^n}(-1)^{f(x)+\omega\cdot x}, $$
where $``\cdot"$ is an inner product on $\gf_2^n$. 
{ The Walsh transform of $f$ can be seen as the \textit{discrete Fourier transform} of the function $(-1)^{f(x)}$ and yields the well-known Parseval's relation \cite{Carlet2010}} :
$$\sum_{\omega\in\gf_{2}^n}W_f^2(\omega)=2^{2n}.$$
The \textit{nonlinearity} of $f$ is defined by
 $$\mathtt{NL}(f) = 2^{n-1} - \frac{1}{2} \max_{\omega\in\gf_{2}^n}|W_f(\omega)|, $$
 where $| r|$ denotes the absolute value of any integer $r$.  According to the Parseval's relation, it is easily seen that the nonlinearity of an $n$-variable Boolean function is upper bounded by $2^{n-1}-2^{n/2-1}$. 
 Boolean functions achieving the maximum nonlinearity are call \emph{bent} functions \cite{Carlet2010,Rothaus1976}, of which the Walsh transform takes only two values $ \pm 2^{n/2}$.

For an $(n,m)$-function $F$ from $\gf_2^n$ to $\gf_2^m$, its \textit{component Boolean function} for a nonzero $v\in \gf_2^m$ is given by 
$$ f_{v} (x) = v\cdot F(x).$$ For any $u\in\gf_{2}^n$ and $v\in\gf_{2}^m$, the Walsh transform of $F$ is defined by that of the component Boolean functions $f_v$, i.e.,
$$W_F(u,v) = \sum_{x\in\gf_{2}^n}(-1)^{u\cdot x + v\cdot F(x)}.$$ Moreover, the nonlinearity of $F$ is defined by the nonlinearities of the component Boolean functions, namely,
$$\mathtt{NL}(F) = \min_{v\in\gf_{2}^m\backslash\{0\}}\mathtt{NL}(v\cdot F).$$

An $(n,m)$-function $F$ is called \textit{vectorial bent}, or shortly \textit{bent} if 
all its component Boolean functions $f_v = v\cdot F(x)$ for each nonzero $v\in \gf_2^m$ are bent. It is well known that $F$ is bent only if $n$ is even and $m\le \frac{n}{2}$. 
Interested readers can refer to \cite{Sihem2016} for more results about bent functions.
For $(n,m)$-functions $F$ with $m\ge n-1$, the Sidelnikov-Chabaud-Vaudenay bound
$$\mathtt{NL}(F)\le 2^{n-1}-\frac{1}{2} \left( 3\cdot 2^n - 2(2^n-1)(2^{n-1}-1)/(2^m-1) -2  \right)^{1/2}$$
gives a better upper bound for nonlinearity than the universal bound \cite{CV1995}. This bound can be achieved 
by the \emph{almost bent} (AB) functions for odd $n$ and $n=m$, where the inequality becomes
$$\mathtt{NL}(F)\le 2^{n-1}-2^{\frac{n-1}{2}}.$$
 It is well known that a function $F$ from $\gf_{2}^n$ to itself is AB if and only if its Walsh transform takes only three values $0, \pm 2^{\frac{n+1}{2}}$ \cite{CV1995}.

A Boolean functions is called \emph{plateaued} if its Walsh transform takes at most three values: $0$ and $\pm \mu$ (where $\mu$ is some positive integer, called the \emph{amplitude} of the plateaued function). It is clear that bent functions are plateaued. Because of Parseval's relation, the amplitude $\mu$ of any plateaued function must be of the form $2^r$ where $r\ge n/2$. An $(n,m)$-function is called \emph{plateaued} if all its component functions $f_v=v\cdot F$, where $v\in\gf_{2}^m\backslash\{0\}$, are plateaued, with possibly different amplitudes. In particular, an $(n,m)$-function $F$ is called \emph{plateaued with single amplitude} if all its component functions are plateaued with the same amplitude. 
It is clear that AB functions are a subclass of plateaued functions with the single amplitude $2^{\frac{n+1}{2}}$. 

\subsection{Differential uniformity}
For an $(n,m)$-function $F$ and any $u\in\gf_2^n\backslash\{0\}$, the function 
$$D_uF(x) = F(x)+F(x+u)$$
is called the derivative of $F$ in direction $u$. The \textit{differential distribution table} (DDT) of $F$ is given by a $2^n\times 2^m$ table, in which the entry for the $(u,v)$ position is given by 
$$ \mathtt{DDT}_F(u,v) =\# \{ x\in\gf_2^n ~|~ D_uF(x) = v \}, $$
where  $u\in\gf_{2}^n$ and $v\in\gf_{2}^m$. The \textit{differential uniformity} of $F$ is defined as 
$$\delta_F=\max_{u\in\gf_{2}^n\backslash\{0\}, v\in\gf_{2}^m}\mathtt{DDT}_F(u,v).$$

Since $D_uF(x)=D_uF(x+u)$ for any $x, u$ in $\gf_2^n$, the entries of DDT are always even and the minimum of differential uniformity of $F$ is $2$. The functions with differential uniformity $2$ are called \emph{almost perfect nonlinear} (APN) functions.

{  
\subsection{The additive autocorrelation of (vectorial) Boolean functions} \label{Subsec-AAC}

In this subsection, we first recall the additive autocorrelation of cryptographic Boolean functions introduced in \cite{ZZ1995}  and extend the definition to vectorial Boolean functions.

\begin{Def} \cite{ZZ1995} Given a Boolean function $f$ on $\gf_{2^n}$. For each $u\in\gf_{2^n}$, the \emph{additive autocorrelation} of the function $f$ at $u$ is defined as 
	$$\Delta_f(u) = \sum_{x\in\gf_{2^n}}(-1)^{f(x)+f(x+u)}.$$ 
	Furthermore, the \emph{absolute indicator} of $f$ is defined as $\Delta_f = \max_{u\in \gf_{2^n}^{*}}|\Delta_f(u)|$ and the  \emph{sum-of-squares} indicator of $f$ is defined by $\nu_f=\sum_{u\in\gf_{2}^n}\Delta_f^2(u)$.
\end{Def}

Similar to the Walsh transform, we define the above criteria of a vectorial Boolean function $F$ in terms of its component Boolean functions and 
use the same names.
\begin{Def}
	\label{AC_Def}
	Let $F$ be an $(n,m)$-function. For any  $u\in\gf_{2}^n$ and $v\in\gf_{2}^m$,  the \emph{additive autocorrelation} of $F$ at $(u,v)$ is defined as
	$$\Delta_F(u,v) = \sum_{x\in\gf_{2}^n}(-1)^{v\cdot (F(x) + F(x+u))},$$
	and the \emph{autocorrelation spectrum} of $F$ is given by
	$$ \Lambda_F := \Big\{ \Delta_F(u,v): u\in\gf_{2}^n\backslash \{ 0 \}, v\in\gf_{2}^m \backslash \{ 0 \}   \Big\}. $$ 
	Moreover, the \emph{absolute indicator} of $F$ is defined as $$\Delta_F =\max_{u\in\gf_{2}^n\backslash\{0\}, v\in\gf_{2}^m\backslash\{0\}}|\Delta_F(u,v)|, $$ 
	and the \emph{sum-of-squares} indicator of $F$ is defined as 
	$$\nu_F=\sum_{u\in\gf_{2}^n,v\in\gf_2^m}\Delta_F^2(u,v),$$
	
\end{Def}
\begin{Rem}
	\label{remark_Delta}
	\emph{From the above definition, it is clear that when $u=0$ or $v=0$, $\Delta_F(u,v)=2^n$. Moreover, for any  $u\in\gf_{2}^n$ and $v\in\gf_{2}^m$, we have
		\begin{equation*}
		\Delta_F(u,v) =  \sum_{\omega\in\gf_{2}^m}(-1)^{\omega\cdot v} \mathtt{DDT}_F(u,\omega). 
		\end{equation*}}
\end{Rem}
}

\subsection{Kloosterman sums}

The Kloosterman sums  have been widely studied for a long time for their own sake as interesting mathematical objects and have recently become the focus of much research. 
Below we recall the classical binary Kloosterman sums and some results that are used in this paper.
 
The Kloosterman sum over $\gf_{2^n}$ is defined as 
$$ K(a) = \sum_{x\in\gf_{2^n}^{*}}(-1)^{\tr_{2^n}\left(\frac{1}{x}+ax\right)}.  $$

\begin{Lemma}
	\label{K_range}
	\cite{LW1990}
	Let $n\ge 3$, for any integer $s\equiv-1\pmod 4$ in the range 
	$$ \left[ -2^{ \frac{n}{2}  +1}, 2^{ \frac{n}{2}  +1}    \right], $$
	there is an element $a\in\gf_{2^n}$ such that $K(a) = s$. 
\end{Lemma}

\begin{Lemma}
	\label{KK}
	\cite{CHZ2007,HZ1999,L2008}
	Let $n\ge 3$. For any $v\in\gf_{2^n}$, $K(v) \equiv -1 \pmod 8$ if $\tr_{2^n}(v) = 0$ and $K(v) \equiv 3 \pmod 8 $ if $\tr_{2^n}(v) = 1$.
\end{Lemma}

\subsection{Differential-Linear Connectivity Table}

Very recently, Bar-On et al. in \cite{DLCT2019} presented the concept of the differential-linear connectivity table (DLCT) of $(n,m)$-functions $F$. 
Due to the isomorphism between vector spaces and finite fields,
the definition of DLCT will be converted to that of polynomials over finite fields.

\begin{Def}
	\cite{DLCT2019}
	Let $F$ be an $(n,m)$-function. The DLCT of $F$ is an $2^n\times 2^m$ table whose rows correspond to input differences to $F$ and whose columns correspond to bit masks of outputs of $F$. Formally, for $u\in\gf_2^n$ and $v\in\gf_2^m$, the DLCT entry $(u, v)$ is 
	$$\mathtt{DLCT}_F(u, v)  = \#  \{ x \in\gf_2^n | v\cdot F(x) = v \cdot F(x+u)  \}   -2^{n-1}.  $$
\end{Def}

Since for any $u\in\gf_{2}^n\backslash\{0\}$, $D_uF(x) = D_uF(x+u)$, $\mathtt{DLCT}_F(u,v)$ must be even. Furthermore, for a given $u\in\gf_{2}^n\backslash\{0\}$, if  $D_uF(x)$ is a $2\ell$ to $1$ mapping, where $\ell$ is a positive integer, then $\mathtt{DLCT}_F(u,v)$ is a multiple of $2\ell$. Moreover, it is trivial that for any $(u, v)\in\gf_2^n\times \gf_{2}^m$, $\left|  \mathtt{DLCT}_F(u, v)\right| \le 2^{n-1} $, and $\mathtt{DLCT}_F(u,v) = 2^{n-1}$ when either $u=0$ or $v=0$. Therefore, 
we only need to focus on the cases for $u\in\gf_{2}^n\backslash\{0\}$ and $v\in\gf_{2}^m\backslash\{0\}$.
In addition, we introduce some relevant definitions of DLCT. 

\begin{Def}
Let $F$ be an $(n,m)$-function. The DLCT spectrum of $F$ is the super set of 
$\mathtt{DLCT}_F(u,v)$ for all nonzero $u$ and $v$, namely,
	$$ \Gamma_F = \Big\{ \mathtt{DLCT}_F(u,v): u\in\gf_{2}^n\backslash \{ 0 \}, v\in\gf_{2}^m \backslash \{ 0 \}   \Big\}, $$ 	
	and the differential-linear uniformity (DLU) of $F$ is defined as 
	$$ \mathtt{DLU}_F := \max\limits_{u\in\gf_{2}^n\backslash\{0\}, v\in\gf_{2}^m\backslash\{0\}} | \mathtt{DLCT}_F(u,v) |. $$
\end{Def}

In \cite{DLCT2019} the authors  gave the follwing relation between the DLCT and the Fourier-Walsh transform of the DDT: $$\mathtt{DLCT}_F(u,v) = \frac{1}{2} \sum_{\omega\in\gf_{2}^m}(-1)^{\omega\cdot v} \mathtt{DDT}_F(u,\omega). $$ 
{This immediately gives  the following connection between the DLCT and the additive autocorrelation of vectorial Boolean functions:
	
\begin{Prop}
		\label{DLCT_Delta}
		Let $F$ be an $(n,m)$-function. Then for any $u\in\gf_{2}^n$ and $v\in\gf_{2}^m$, the additive autocorrelation of $F$ at $(u,v)$ is twice the value of the DLCT of $F$ at the same position $(u,v)$, i.e., 
		\begin{equation}
		\label{DLCT_W}
		\mathtt{DLCT}_F(u,v) = \frac{1}{2}\Delta_F(u,v). 
		\end{equation} 
		Moreover, $\mathtt{DLU}_F = \frac{1}{2}\Delta_F$.
\end{Prop}  }

\section{Some characterizations and properties of DLCT}
\label{Characterizations}

{In this section, we give some characterizations and properties of DLCT of vectorial Boolean functions from the viewpoint of the additive autocorrelation introduced in Subsection \ref{Subsec-AAC}. }

\subsection{Characterizations of DLCT by means of the Walsh transform}
In this subsection, we give some characterizations of DLCT by the Walsh transform. 
\begin{Prop}
	\label{DLCT_WT}
	 Let $F$ be an $(n,m)$-function.	Then for any $u\in\gf_{2}^n$ and $v\in\gf_{2}^m$, 
	 \begin{enumerate}[(1)]
	 	\item $$\mathtt{DLCT}_F(u,v) = \frac{1}{2} W_{D_uF}(0, v);$$
  	\item  	\begin{equation}
\label{DW}
\mathtt{DLCT}_F(u,v) = \frac{1}{2^{n+1}}\sum_{\omega\in\gf_2^n}(-1)^{u\cdot \omega} W_F(\omega,v)^2;
\end{equation}
Moreover, we have 
$$\sum_{u\in\gf_{2}^n} 	\mathtt{DLCT}_F(u,v) = \frac{1}{2}W_F(0,v)^2; $$
\item 	\begin{equation}
\label{D2W4}
\sum_{u\in\gf_{2}^n}\mathtt{DLCT}_F(u,v)^2 = \frac{1}{2^{n+2}}\sum_{\omega\in\gf_{2}^n}W_F(\omega,v)^4. 
\end{equation} 
	 \end{enumerate}
\end{Prop}

\begin{proof}
\begin{enumerate}[(1)]
	\item This item can be easily verified according to the definition. 
		\item 	According to the definition, for any $\omega\in\gf_{2}^n$,
	\begin{eqnarray*}
		W_F(\omega,v)^2 &=& \sum_{x\in\gf_{2}^n}(-1)^{\omega\cdot x + v\cdot F(x)}\sum_{y\in\gf_{2}^n}(-1)^{\omega\cdot y + v\cdot F(y)} \\
		&=&\sum_{x,y\in\gf_2^n}(-1)^{\omega\cdot(x+y)+ v\cdot (F(x) + F(y))} \\
		&=&\sum_{x,u\in\gf_2^n}(-1)^{\omega\cdot u + v\cdot (F(x) + F(x+u))} \\
		&=&\sum_{u\in\gf_2^n}(-1)^{\omega\cdot u}\sum_{x\in\gf_2^n}(-1)^{v\cdot (F(x) + F(x+u))}\\
		&=&\sum_{u\in\gf_2^n}(-1)^{\omega\cdot u}\Delta_F(u,v).
	\end{eqnarray*}
	
	From the Inverse Discrete Fourier Transform, we have 
	$$\Delta_F(u,v) = \frac{1}{2^n}\sum_{\omega\in\gf_2^n}(-1)^{\omega\cdot u} W_F(\omega,v)^2. $$
	Therefore, Eq. (\ref{DW})  holds with Eq. (\ref{DLCT_W}). Moreover, we have 
	\begin{eqnarray*}
		&&	\sum_{u\in\gf_2^n}\mathtt{DLCT}_F(u,v) \\
		&=& \frac{1}{2^{n+1}} \sum_{\omega\in\gf_{2}^n} W_F(\omega,v)^2\sum_{u\in\gf_2^n}(-1)^{\omega\cdot u}\\
		&=& \frac{1}{2} W_F(0,v)^2.
	\end{eqnarray*}
	\item For any $u\in\gf_{2}^n$ and $v\in\gf_{2}^m$, 
	\begin{eqnarray*}
		\mathtt{DLCT}_F(u,v)^2 & = & \frac{1}{2^{2n+2}} \sum_{\omega_1\in\gf_{2}^n} (-1)^{u\cdot \omega_1}W_F(\omega_1,v)^2\sum_{\omega_2\in\gf_{2}^n} (-1)^{u\cdot \omega_2}W_F(\omega_2,v)^2\\
		&=& \frac{1}{2^{2n+2}}\sum_{\omega_1,\omega_2\in\gf_{2}^n} (-1)^{u\cdot\left( \omega_1+\omega_2\right)}W_F(\omega_1,v)^2W_F(\omega_2,v)^2.
	\end{eqnarray*}
	Furthermore, 
	\begin{eqnarray*}
		\sum_{u\in\gf_{2}^n}\mathtt{DLCT}_F(u,v)^2  & = & \frac{1}{2^{2n+2}} \sum_{u,\omega_1,\omega_2\in\gf_{2}^n} (-1)^{u\cdot\left( \omega_1+\omega_2\right)}W_F(\omega_1,v)^2W_F(\omega_2,v)^2 \\
		&=& \frac{1}{2^{2n+2}} \sum_{\omega_1,\omega_2\in\gf_{2}^n} W_F(\omega_1,v)^2W_F(\omega_2,v)^2 \left(\sum_{u\in\gf_{2}^n}   (-1)^{u\cdot\left( \omega_1+\omega_2\right)} \right) \\
		&=& \frac{1}{2^{n+2}}\sum_{\omega\in\gf_{2}^n}W_F(\omega,v)^4.
	\end{eqnarray*}
\end{enumerate}
	The proof is completed.
\end{proof}

{
\begin{Rem}
		\emph{ It should be noted that the authors of \cite{GK2004} and \cite{ZZ1995} obtained the relations in Eq. (\ref{DW}) and Eq. (\ref{D2W4}) for Boolean functions.
		Here we generalize those results to vectorial Boolean functions. }
\end{Rem} 
}

\subsection{Characterizations of DLCT by means of the DDT}
We in this subsection consider the characerizations of DLCT by means of the DDT. 
\begin{Prop}
	\label{DLCT_DDT}
	Let $F$ be an $(n,m)$-function. Then	for any $u\in\gf_{2}^n$ and $v\in\gf_{2}^m$,
	\begin{enumerate}[(1)]
		\item  $$\sum_{v\in\gf_{2}^m} \mathtt{DLCT}_F(u,v) = 2^{m-1} \mathtt{DDT}_F(u,0);  $$ 
		moreover, $$\sum_{u\in\gf_2^n,v\in\gf_2^m}\mathtt{DLCT}_F(u,v)= 2^{m+n-1}.$$ 
		In particular, when $m=n$ and $F$ permutes $\gf_{2}^n$,  $\sum_{v\in\gf_{2}^n} \mathtt{DLCT}_F(u,v) = 0;$
		\item 	
		\begin{equation}
		\label{D2DDT2}
		\sum_{v\in\gf_2^m} \mathtt{DLCT}_F(u,v)^2 = 2^{m-2} \sum_{\omega\in\gf_2^m} \mathtt{DDT}_F(u,\omega)^2.
		\end{equation}
	\end{enumerate} 
\end{Prop}

\begin{proof}
	\begin{enumerate}[(1)]
		\item According to Eq. (\ref{DLCT_W}),
		\begin{eqnarray*}
			&   &  \sum_{v\in\gf_{2}^m} \mathtt{DLCT}_F(u,v) = \frac{1}{2} \sum_{v\in\gf_{2}^m} \Delta_F(u,v) \\ 
			& = & \frac{1}{2} \sum_{v\in\gf_{2}^m}\sum_{x\in\gf_{2}^n} (-1)^{v\cdot (F(x)+F(x+u))} \\
			& = & 2^{m-1} \#\{ x\in\gf_{2}^n | F(x)+F(x+u) = 0  \} =  2^{m-1} \mathtt{DDT}_F(u,0). 
		\end{eqnarray*}
		In particular, when  $m=n$ and $F$ permutes $\gf_{2}^n$, $\mathtt{DDT}_F(u,0)=0$ and thus $\sum_{v\in\gf_{2}^n} \mathtt{DLCT}_F(u,v) = 0.$
		\item It holds since for any $u\in\gf_{2}^n$,
		\begin{eqnarray*} 
			& &	\sum_{v\in\gf_2^m}\mathtt{DLCT}_F(u,v)^2 = \frac{1}{4}	\sum_{v\in\gf_2^m}\Delta_F(u,v)^2 \\
			& =& \frac{1}{4} \sum_{v\in\gf_2^m}\sum_{x,y\in\gf_2^n} (-1)^{v\cdot (F(x)+F(x+u)+F(y)+F(y+u))} \\
			&=& \frac{1}{4}\cdot 2^m \cdot \# \left\{ (x,y)\in\gf_{2}^n\times\gf_{2}^n | F(x)+F(x+u)+F(y)+F(y+u) =0  \right\}\\
			&=& 2^{m-2}\cdot \sum_{\omega\in\gf_{2}^m} \mathtt{DDT}_F(u,\omega)^2.
		\end{eqnarray*}
	\end{enumerate}	
\end{proof}

\subsection{Characterization of bounds on the DLU}
Similar to other cryptographic criteria, it is interesting and important to know 
how ``good" the DLU of a vectorial Boolean function could be. 
It is clear that the DLU of any $(n,m)$-functions is upper bounded by $2^{n-1}$. The following theorem characterizes the lower bound on DLU of any vectorial Boolean functions.

\begin{Th}
	Let $F$ be an $(n,m)$-function, where $m\ge n-1$. Then 
	\begin{equation}
	\label{DLU_bound}
	\mathtt{DLU}_F\ge \sqrt{\frac{2^{m+n+1}-2^{2n}}{4(2^m-1)}}. 
	\end{equation}
\end{Th}
\begin{proof}
	It is well known that for any functions $F$ from $\gf_{2}^n$ to $\gf_{2}^m$, $$\sum_{w\in\gf_{2}^n,v\in\gf_{2}^m}W_F(w,v)^4 \ge 2^{n+m}\left(3\cdot 2^{2n} - 2\cdot 2^n\right)$$ with equality attained if and only if $F$ is APN. It follows from Eq. (\ref{D2W4}) that
	\begin{eqnarray*}
		\sum_{u\in\gf_{2}^n, v\in\gf_{2}^m}	\mathtt{DLCT}_F(u,v)^2 &=& \frac{1}{2^{n+2}}\sum_{b\in\gf_{2}^n,v\in\gf_{2}^m}W_F(b,v)^4 \\
		&\ge& \frac{1}{2^{n+2} } \cdot  2^{n+m}\left(3\cdot 2^{2n} - 2\cdot 2^n\right) \\
		& = & 3\cdot 2^{2n+m-2}-2^{n+m-1},
	\end{eqnarray*}
	or equivalently,
	\begin{equation}
	\label{bound}
	\sum_{u\in\gf_{2}^n\backslash\{0\}, v\in\gf_{2}^m\backslash\{0\}}	\mathtt{DLCT}_F(u,v)^2 \ge 2^{2n+m-1}+2^{2n-2}-2^{3n-2}-2^{n+m-1}.
	\end{equation}
	Since $	\mathtt{DLU}_F := \max_{u\in\gf_{2}^n\backslash\{0\}, v\in\gf_{2}^m\backslash\{0\}} \left| 	\mathtt{DLCT}_F(u,v)  \right|$, we have 
	\begin{eqnarray*}
		(2^n-1)(2^m-1) \mathtt{DLU}_F^2 &=& \sum_{b\in\gf_{2}^n\backslash\{0\}, v\in\gf_{2}^m\backslash\{0\}} \mathtt{DLU}_F^2  \\
		&\ge & \sum_{b\in\gf_{2}^n\backslash\{0\}, v\in\gf_{2}^m\backslash\{0\}}\mathtt{DLCT}_F(b,v)^2 \\
		&\ge & 2^{2n-2}\left(2^n-1\right)\left(2^{m-n+1}-1\right).
	\end{eqnarray*}
The desired conclusion follows from the above inequality.
\end{proof}
The condition $m \ge n-1$ is assumed in the above theorem to make non-negative
the expression located under the square root. Note that for $m = n- 1,$ 
it leads to a trivial lower bound; and for $m=n$, it gives a non-trivial lower bound
$$
	\mathtt{DLU}_F\ge \frac{2^{n-1}}{\sqrt{2^n-1}}>2^{\frac{n}{2}-1}.
$$
%
In the case that $n$ is even, since $\mathtt{DLU}_F$ must be even, we have 
\begin{Cor}
	\label{bound_F_2n}
	Let $F$ be a function from $\gf_{2}^n$ to itself, where $n$ is even. Then
	\begin{equation}
	\label{nn_DLU_bound}
	\mathtt{DLU}_F \ge 2^{\frac{n}{2}-1}+2.
	\end{equation}
\end{Cor}

\begin{Rem}
	\emph{ Let $F$ be a function from $\gf_{2}^4$ to itself. Then $\mathtt{DLU}_F\ge 4$. Moreover, it is direct to check by Magma that $\mathtt{DLU}_F=4$, where $F$ is the inverse function over $\gf_{2}^4.$ However, we think the bound can be improved from the experiment result by Magma.}
\end{Rem}

\begin{Rem}
	\label{Bent}	
	\emph{Let $F$ be an $(n,m)$-function. According to the definition, for any $u\in\gf_{2}^n\backslash\{0\}, v\in\gf_{2}^m\backslash\{0\}$, $\mathtt{DLCT}_F(u,v) = \Delta_F(u,v)=0$ if and only if   $\#  \{ x \in\gf_2^n | v\cdot F(x) = v \cdot F(x+u)  \}  = 2^{n-1}$, which means that $F(x)$ is a  bent function. Therefore, the DLCT spectrum $\Gamma_F = \{0\}$ if and only if $F$ is a bent function. It is well known that $F$ is bent only if $n$ is even and $m\le \frac{n}{2}$.  Thus for the case $n$ is even and $m\le \frac{n}{2}$, the DLU of any $(n,m)$-function $F$ satisfies $\mathtt{DLU}_F\ge0$ and the equality is attainable.}
\end{Rem}

\subsection{DLCT spectrum and DLU under three equivalence relations}

Let $n, m$ be two positive integers. There are several equivalence relations of functions from $\gf_p^n$ to $\gf_p^m$  and they play vital roles in constructing functions with good properties, like AB and APN functions \cite{BCP2006}. In this subsection, we first recall three equivalence relations, i.e., affine, EA and CCZ. Then we study the DLCT and relative concepts with the perspective of equivalence relations.

\begin{Def}
	\cite{Bud2014}
	Let $p$ be prime and $n,m$ be positive integers. Two functions $F$ and $F^{'}$ from $\gf_p^n$ to $\gf_p^m$ are called 
	\begin{enumerate}
		\item {affine equivalent ({or} linear equivalent) }  if $F^{'} = A_1\circ F\circ A_2$, where the mappings $A_1$ and $A_2$ are affine (resp. linear) permutations of $\gf_p^m$ and $\gf_p^n$, respectively;
		\item { extended affine equivalent} (EA equivalent) if $F^{'} =  A_1\circ F \circ A_2 + A$, where the mappings $A : \gf_p^n\to\gf_p^m, A_1 : \gf_p^m\to\gf_p^m,  A_2 : \gf_p^n\to\gf_p^n$ are affine and where $A_1, A_2$ are permutations;
		\item { Carlet-Charpin-Zinoviev equivalent } (CCZ equivalent) if for some affine permutation $\mathcal{L}$ over $\gf_p^n \times \gf_p^m$, the image of the graph of $F$ is the graph of $F^{'}$, that is $\mathcal{L}(G_F)=G_{F^{'}}$, where $G_F=\{ (x,F(x)) | x\in\gf_p^n  \}$ and $G_{F^{'}}= \{ (x,F^{'}(x)) | x\in\gf_p^n  \} $.
	\end{enumerate}
\end{Def}

It is known that affine equivalence is a particular case of EA equivalence. EA equivalence is also a particular case of CCZ equivalence, and every permutation is CCZ equivalent to its compositional inverse.  For two important properties of cryptographic functions, i.e., differential uniformity and nonlinearity, they are both CCZ equivalent invariants.  However, as we will show in this subsection, the Differential-Linear uniformity is only an EA equivalent invariant, instead of a CCZ one. Moreover, the DLCT spectrum is only an affine equivalent invariant. 

\begin{Th}
	Assume $F$ and $F^{'}$ from $\gf_{2}^n$ to $\gf_{2}^m$ are EA equivalent. Then $\mathtt{DLU}_F = \mathtt{DLU}_{F^{'}}$. Moreover, if  $F$ and $F^{'}$ from $\gf_{2}^n$ to $\gf_{2}^m$ are affine equivalent, $\Lambda_F = \Lambda_{F^{'}}$, $\Gamma_F = \Gamma_{F^{'}}$.
\end{Th}

\begin{proof}
	Since $F$ and $F^{'}$ are EA equivalent, there exist affine mappings  $A : \gf_2^n\to\gf_2^m, A_1 : \gf_2^m\to\gf_2^m,  A_2 : \gf_2^n\to\gf_2^n$, where $A_1, A_2$ are permutations, such that $F^{'} =  A_1\circ F \circ A_2 + A$. Assume that the linear parts of $A, A_1, A_2$ are $L, L_1, L_2$ respectively. Then for any  $u\in\gf_{2}^n\backslash\{0\}$ and $v\in\gf_{2}^m\backslash\{0\}$,
	\begin{eqnarray*}
		\Delta_{F^{'}}(u,v) & = & \sum_{x\in\gf_2^n}(-1)^{v\cdot \left( F^{'}(x) + F^{'}(x+u) \right)} \\
		&=& \sum_{x\in\gf_2^n} (-1)^{ v\cdot  \left( A_1\circ F \circ A_2(x) + A(x) +   A_1\circ F \circ A_2(x+u) + A(x+u)     \right)  } \\
		&=& (-1)^{v\cdot L(u)} \sum_{x\in\gf_2^n} (-1)^{v\cdot  \left( A_1\circ F \circ A_2(x) +   A_1\circ F \circ A_2(x+u)   \right)} \\
		&=&  (-1)^{v\cdot L(u)} \sum_{x\in\gf_2^n} (-1)^{ v \cdot L_1 \left(  F \circ A_2(x) + F \circ A_2(x+u) \right)   } \\
		&=& (-1)^{v\cdot L(u)} \sum_{x\in\gf_2^n} (-1)^{L_1^{\mathrm{T}}(v) \cdot \left(  F \circ A_2(x) + F \circ A_2(x+u) \right)  }\\
		&=& (-1)^{v\cdot L(u)}  \sum_{y\in\gf_2^n} (-1)^{L_1^{\mathrm{T}}(v) \cdot \left( F(y) + F  \left( y+L_2(u)  \right)   \right) } \\
		&=& (-1)^{v\cdot L(u)} 	\Delta_{F}(L_2(u),L_1^{\mathrm{T}}(v) ),
	\end{eqnarray*}
	where $L_1^{\mathrm{T}}$ denotes the linear mapping whose corresponding matrix is the transpose  of the corresponding matrix of $L_1$.  Hence, $\mathtt{DLU}_F = \mathtt{DLU}_{F^{'}}$. Moreover, when $F$ and $F^{'}$ from $\gf_{2}^n$ to $\gf_{2}^m$ are affine equivalent, namely, $A=0$, we have 
	$$ 	\Delta_{F^{'}}(u,v) =  \Delta_{F}(L_2(u),L_1^{\mathrm{T}}(v) ) $$
	and thus $\Lambda_F = \Lambda_{F^{'}}$, naturally $\Gamma_F = \Gamma_{F^{'}}$.
\end{proof}

\begin{example}
	There are two examples to claim that the Differential-Linear uniformity is not an invariant of CCZ equivalence and the DLCT spectrum is not an EA equivalence.
	\begin{enumerate}
		\item  Let $F(x)=x^{13}\in\gf_{2^6}[x]$ be a permutation over $\gf_{2^6}$ and $F^{-1}(x)=x^{34}$. It is clear that $F(x)$ and $F^{-1}(x)$ are CCZ equivalent. However, by Magma, we obtain $\Gamma_F = \{ - 16,-8, 0, 8, 16  \}$ and $ \Gamma_{F^{-1}} = \{ -32, 0, 32  \}.$ Naturally, $\mathtt{DLU}_F = 16$ and $\mathtt{DLU}_{F^{-1}} = 32$;
		\item Let $F(x) = \frac{1}{x}\in\gf_{2^7}[x]$ and $F^{'}(x) = \frac{1}{x}+x$. Then $F(x)$ and $F^{'}(x)$ are EA equivalent. However, by Magma, we obtain $\Gamma_F = \{ -12, -8,-4, 0, 4, 8 \}$ while $\Gamma_{F^{'}} = \{ -12, -8,-4, 0, 4, 8, 12\}$. 
	\end{enumerate}
	
\end{example}

In \cite{LP2007}, the authors classified all optimal S-boxes (permutations) over $\gf_{2}^4$ having best differential uniformity and nonlinearity (both $4$) up to affine equivalence and found that there are only $16$ different optimal S-boxes, see Table \ref{4Sbox}. Based on the classification of optimal S-boxes, we obtain all possibilities of the DLCT spectrum and DLU of optimal S-boxes, see Table \ref{4Autocorrelation}. 

\begin{table}[!htbp]
	\caption{Representatives for all $16$ classes of optimal $4$ bit Sboxes }
	\centering
	\label{4Sbox}
	\begin{tabular}{cc}	
		\hline
		$F_0$ ~&~  $ 0, 1, 2, 13, 4, 7, 15, 6, 8, 11, 12, 9, 3, 14, 10, 5 $ \\
		$F_1$ ~&~	$0, 1, 2, 13, 4, 7, 15, 6, 8, 11, 14, 3, 5, 9, 10, 12 $ \\
		$F_2$	~&~ $0, 1, 2, 13, 4, 7, 15, 6, 8, 11, 14, 3, 10, 12, 5, 9 $ \\
		$F_3$	~& ~   $0, 1, 2, 13, 4, 7, 15, 6, 8, 12, 5, 3, 10, 14, 11, 9$    \\
		$F_4$	 ~&~	$ 0, 1, 2, 13, 4, 7, 15, 6, 8, 12, 9, 11, 10, 14, 5, 3 $ \\
		$F_5$	 ~&~  $0, 1, 2, 13, 4, 7, 15, 6, 8, 12, 11, 9, 10, 14, 3, 5 $  \\
		$F_6$ ~&~  $0, 1, 2, 13, 4, 7, 15, 6, 8, 12, 11, 9, 10, 14, 5, 3$       \\
		$F_7$ ~&~	$0, 1, 2, 13, 4, 7, 15, 6, 8, 12, 14, 11, 10, 9, 3, 5$ \\
		$F_8$	~&~ $0, 1, 2, 13, 4, 7, 15, 6, 8, 14, 9, 5, 10, 11, 3, 12$ \\
		$F_9$	~& ~  $ 0, 1, 2, 13, 4, 7, 15, 6, 8, 14, 11, 3, 5, 9, 10, 12 $ \\
		$F_{10}$	 ~&~ $0, 1, 2, 13, 4, 7, 15, 6, 8, 14, 11, 5, 10, 9, 3, 12$	 \\
		$F_{11}$	 ~&~ $0, 1, 2, 13, 4, 7, 15, 6, 8, 14, 11, 10, 5, 9, 12, 3$  \\
		$F_{12}$	~&~ $0, 1, 2, 13, 4, 7, 15, 6, 8, 14, 11, 10, 9, 3, 12, 5$ \\
		$F_{13}$	~& ~  $0, 1, 2, 13, 4, 7, 15, 6, 8, 14, 12, 9, 5, 11, 10, 3$\\
		$F_{14}$	 ~&~ $0, 1, 2, 13, 4, 7, 15, 6, 8, 14, 12, 11, 3, 9, 5, 10$	 \\
		$F_{15}$	 ~&~ $0, 1, 2, 13, 4, 7, 15, 6, 8, 14, 12, 11, 9, 3, 10, 5$ \\	
		\hline
	\end{tabular}	
\end{table}

\begin{table}[!htbp] 
	\caption{Autocorrelation spectra and DLU of $F_i$ for $ 0\le i\le 15$}	\label{4Autocorrelation}
	\centering
	\begin{tabular}{ccc}	
		\hline
		$F_i$ & $i = 3\sim 7,11\sim 13$  & $ i = 0 \sim 2,8 \sim 10,14,15$   \\
		\hline
		DL-Walsh spectrum &	$\{-8,0,8 \}$ & $ \{-16,-8,0,8,16 \}$   \\
		DLCT spectrum &	$\{-4,0,4 \}$ & $ \{-8,-4,0,4,8 \}$ \\
		DLU &	$4$ & $8$  \\
		\hline
	\end{tabular}	
\end{table}

\subsection{DLCT of monomials} 

We in the subsection consider the DLCT of monomials over the finite field $\gf_{2^n}$. 

\begin{Prop}
	\label{prop_monomials}
	Let $F(x) = x^d \in \gf_{2^n}[x]$. Then $$\Gamma_F := \left\{ \frac{1}{2} \Delta_F(1,v): v\in\gf_{2^n}^{*} \right\}.$$ Moreover, if $\gcd\left(d,2^n-1\right)=1$, then 
	$$\Gamma_F := \left\{ \frac{1}{2} \Delta_F(u,1): u \in \gf_{2^n}^{*}  \right\}.$$
\end{Prop}

\begin{proof}
	For any $u,  v\in\gf_{2^n}^{*},$ we have
	\begin{eqnarray*}
		\Delta_F(u,v) &=& \sum_{x\in\gf_{2^n}}(-1)^{\tr_{2^n}\left( v(F(x)+F(x+u))\right)  }\\
		&=& \sum_{x\in\gf_2^n}(-1)^{ \tr_{2^n}\left( v\left(  x^d+(x+u)^d \right) \right) }\\
		&=& \sum_{x\in\gf_{2^n}}(-1)^{\tr_{2^n}\left(vu^d\left( \left(\frac{x}{u}\right)^d +   \left(\frac{x}{u}+1\right)^d \right)   \right)}\\
		&=& \Delta_F\left(1, vu^d \right).
	\end{eqnarray*}
	Moreover, if $\gcd\left( d,2^n-1 \right)=1$, then for any $v\in\gf_{2^n}^{*}$, there exists a unique element $u\in\gf_{2^n}^{*}$ such that $v=u^d$. Furthermore, 
	\begin{eqnarray*}
		\Delta_F(1,v) &=& \sum_{x\in\gf_{2^n}}(-1)^{\tr_{2^n}\left(v\left(x^d+(x+1)^d\right)\right)} \\
		&=& \sum_{x\in\gf_{2^n}}(-1)^{\tr_{2^n}\left((ux)^d+(ux+u)^d\right)}\\
		&=& \sum_{y\in\gf_{2^n}}(-1)^{\tr_{2^n}\left(y^d+(y+u)^d\right)}\\
		&=& \Delta_F(u,1).
	\end{eqnarray*}
	The desired conclusion follows from Proposition \ref{DLCT_W}.
\end{proof}

\begin{Rem}
	\emph{Let $F(x) = x^d \in \gf_{2^n}[x]$ with $\gcd\left(d,2^n-1\right)=1$. Then from Proposition \ref{prop_monomials}, we have $\Gamma_F := \left\{ \frac{1}{2} \Delta_F(u,1): u \in \gf_{2^n}^{*}  \right\}.$ In fact, for any $u\in\gf_{2^n}^{*}$, $\Delta_F(u,1) = \sum_{x\in\gf_{2^n}}(-1)^{\tr_{2^n}\left(x^d\right)+ \tr_{2^n}\left( (x+u)^d\right)   }$, which is indeed the additive autocorrelation of the function $\tr_{2^n}\left(x^d\right)$ at $u$. Moreover, the DLU of $F=x^d$ in this case, i.e., $\max_{u\in \gf_{2^n}^{*}}|\Delta_F(u,1)|$ is the absolute indicator  of $\tr_{2^n}\left(x^d\right)$. }	

\emph{	In \cite{ZZ1995}, the authors proved that if $f$ is a non-bent cubic Boolean function on $\gf_{2^n}$. Then  the absolute indicator of $f$ satisfies $\Delta_f\ge 2^{\frac{n+1}{2}}$. Thus if $\tr_{2^n}\left(x^d\right)$ is non-bent cubic, then the DLU of $F=x^d$ satisfies $\mathtt{DLU}_F \ge 2^{\frac{n-1}{2}}$, which is better than the bound in Corollary \ref{bound_F_2n}. 
	In addition, it was conjectured  in \cite{ZZ1995} that the absolute indicator of any $n$-variable balanced Boolean function is lower bounded by $2^{\frac{n+1}{2}}$. Although
	this conjecture was disproved for even $n$, it is believed to be true for odd integer $n$. 
	With the relation between DLU and the absolute indicator, we propose the following optimality of vectorial Boolean functions with respect to the DLCT property.}
			
\end{Rem}

\begin{Def}\label{Def-OptDLU} For an odd integer $n$, an $(n, n)$-function $F$  is said to be optimal with respect to the DLCT 
	if its DLU is equal to $2^{\frac{n-1}{2}}$.
\end{Def}
In the next subsection, we will show that this lower bound can be achieved by certain Kasami-Welch APN monomials.

\subsection{{DLCT of Plateaued, AB and APN functions}}
APN and AB functions provide optimal resistance against differential attack and linear attack, respectively. Many researchers have considered some another properties of APN and AB functions, see \cite{Bud2014}. 
This subsection will investigate the DLCT of these optimal functions. We start with a general result for plateaued functions.

\begin{Prop}
	\label{Prop_AB}
	Assume an $(n,m)$-function $F$ is  plateaued. For $v\in\gf_2^m\backslash\{0\}$, we denote the amplitudes by $2^{r_v}$ and define a  dual
	Boolean function of $f_v$ as
	\begin{equation}
\label{gv}
\widetilde{f}_v(b) = 	\left\{
\begin{array}{lr}
1,  &~  \text{if}~ W_{f_v}(b) \neq0, \\
0, &~ \text{if}~ W_{f_v}(b)=0.
\end{array}
\right.
\end{equation}	
	Then 
	$$\mathtt{DLCT}_F(u,v) = - 2^{2r_v-n-2}W_{\widetilde{f}_v}(u).$$
	Further, when $F$ is an AB function from $\gf_2^n$ to itself, namely, $r_v=2^{\frac{n+1}{2}}$ for any $v\in\gf_{2}^n\backslash\{0\}$, $$\mathtt{DLCT}_F(u,v) = -\frac{1}{2} W_{\widetilde{f}_v}(u). $$
\end{Prop}

\begin{proof}
	According to Eq. (\ref{DW}), we have 
	\begin{eqnarray*}
		\mathtt{DLCT}_F(u,v) &=& \frac{1}{2^{n+1}}\sum_{\omega\in\gf_2^n}(-1)^{u\cdot \omega} W_F(\omega,v)^2 \\
		&=& 2^{2r_v-n-1}\sum_{\omega\in\gf_2^n}(-1)^{u\cdot \omega} \widetilde{f}_b(\omega)\\
		&=& 2^{2r_v-n-1}\sum_{\omega\in\gf_2^n}\left( \frac{1}{2}\left(1-(-1)^{\widetilde{f}_v(\omega)}\right)  \right)(-1)^{u\cdot \omega}\\
		&=& -2^{2r_v-n-2}\sum_{\omega\in\gf_2^n}(-1)^{\widetilde{f}_v(\omega)+u\cdot \omega}\\
		&=& -2^{2r_v-n-2}W_{\widetilde{f}_v}(u).
	\end{eqnarray*}
	Particularly, when $F$ is an AB function, i.e., $r_v=\frac{n+1}{2}$ for any $v\in\gf_2^m\backslash\{0\}$, it is clear that $\mathtt{DLCT}_F(u,v) = -\frac{1}{2} W_{\hat{g}_v}(u).$
\end{proof}

\begin{Rem}
	Proposition \ref{Prop_AB} is inspired by and a generalization of the result in \cite{GK2004}, where the authors investigated the additive autocorrelation of 
		a plateaued Boolean function $f$ in terms of its dual function, and proposed several families of $n$-variable Boolean functions with absolute 
		indicator $2^{(n+1)/2}$.
\end{Rem}

As a natural generalization of Theorem 5 in \cite{GK2004}, the following proposition gives the DLCT of the Kasami-Welch APN monomials for certain parameters.
\begin{Prop}\label{prop-KasamiAPN} Let $n$ be an odd integer coprime to 3, $3k\equiv 1 \pmod{n}$ and $d=2^{2k}-2^k+1$.
Then the Kasami-Welch APN function $F(x) = x^d$ has the DLCT spectrum as $\{0, \pm 2^{(n-1)/2}\}$.
\end{Prop}

\begin{proof} For any nonzero $v$ in $\gf_{2^n}$, it follows from
	the proof of Proposition \ref{prop_monomials} and the fact $\gcd(d, 2^n-1)=1$, it follows that
	$$\Delta_{F}(u, v) = \Delta_{f_v}(u) = \Delta_{f_1}(uv^d),$$
   where $f_v$ is the component function of $F$ given by $f_v(x)=\tr_{2^n}(vF(x))$.
  According to the Walsh spectrum of the function $\tr_{2^n}(x^d)$ \cite{Dillon}, the dual function
  of $f_v$ can be given by
   $
 	\tilde{f}_v(x) = \tr_{2^n}(v_1x^{2^k+1}),
   $ where $v_1=(1/v)^{\frac{2^k+1}{d}}$. This is a Gold-like function and it's well-known that its Walsh spectrum is $\{0, \pm 2^{(n+1)/2}\}$.
   The desired conclusion follows from Proposition \ref{Prop_AB}.
\end{proof}


From the known result in the literature, it appears that optimal $(n, n)$-functions with respect to DLCT are rare objects, which is also confirmed by experimental results for small integer $n$.  Below we propose an open problem for such functions.

\begin{Prob}
	For an odd integer $n$, are there $(n, n)$-functions $F$ other than the Kasami-Welch APN functions that have $\mathtt{DLU}_F=2^{(n-1)/2}$ ?
\end{Prob}

Similar to the AB functions, the DLCT of APN functions can also be expressed by certain Boolean functions as follows.

\begin{Prop}
	\label{Prop_APN}
	Let $F$ be an APN function from $\gf_{2^n}$ to itself. For any nonzero $u\in \gf_{2^n}$, 
	let $D_u = \{F(x+u)+F(x)\; |\; x \in \gf_{2^n}\}$ and 
	define a Boolean function 
	\begin{equation}
	\label{fu}
	\bar{f}_u(x) = 	\left\{
	\begin{array}{ll}
	1,  &~  \text{if}~ x \in D_u, \\
	0, &~ \text{if}~ x\in\gf_{2^n} \backslash D_u.
	\end{array}
	\right.
	\end{equation} Then the DLCT of $F(x)$ satisfies
	$$\mathtt{DLCT}_F(u,v) = -\frac{1}{2} W_{\bar{f}_u}(v).$$
\end{Prop}
\begin{proof} Since the APN function $F(x)$ leads to a $2$-to-$1$ derivative function $D_uF(x)$ at any nonzero $u$,  we known that $D_u$ has $2^{n-1}$ elements. Then,
	\begin{eqnarray*}
		\Delta_F(u,v) &=& \sum_{x\in\gf_{2^n}}(-1)^{\tr_{2^n}(v(F(x+u)+F(x)))} \\
		&=& 2\sum_{y\in D_u} (-1)^{\tr_{2^n}(vy)}\\
		&=& \sum_{y\in D_u}(-1)^{\tr_{2^n}(vy)} - \sum_{y\in \gf_{2^n} \backslash D_u}(-1)^{\tr_{2^n}(vy)}\\
		&=& -\sum_{y\in\gf_{2^n}}(-1)^{\bar{f}_u(y)+\tr_{2^n}(vy)}\\
		&=& -W_{\bar{f}_u}(v).
	\end{eqnarray*}
	The desired conclusion immediately follows from Eq. (\ref{DLCT_W}).
\end{proof}
\mkq{
%
%
%
%
}

\begin{Rem}
	\emph{From Propositions  \ref{Prop_AB} and \ref{Prop_APN}, we see that the values of DLCT of APN and AB functions are given by the Walsh transform of two families of Boolean functions $\widetilde{f}_u$ in \eqref{gv} and $\bar{f}_v$ in (\ref{fu}).  It is easily seen that for APN and AB functions, both these two families of functions are balanced. 
		It is well known that a  Boolean function is balanced if and only if its Walsh transform vanishes at the zero element. Therefore, for any balanced Boolean function $f$, together with the Parseval's relation, the Walsh transform of $f$ satisfies $$|W_f(\omega)|\ge 2^n\sqrt{\frac{1}{2^n-1}}.$$ Thus, the DLU of any APN and AB functions $F$ over $\gf_2^n$ is lower bounded by  $2^{n-1}\sqrt{\frac{1}{2^n-1}}$, which is consistent with the bound in Corollary \ref{bound_F_2n}. }
\end{Rem}

\section{DLCT spectra and DLU of some special polynomials}
\label{polynomials}

In this section, we mainly consider some special polynomials with low differential uniformity. Explicitly, we consider the DLCT spectra and DLU of the inverse, Gold and Bracken-Leander functions. In addition, the DLU of all quadratic polynomials can be determined.  We divide our results into three subsections according to their algebraic degrees.

\subsection{The inverse function}
\begin{Th}
	Let $F(x) = \frac{1}{x}\in\gf_{2^n}[x]$. Then  $$ \Gamma_F =  \left\{ \frac{ K\left(v \right) - 1}{2} + (-1)^{\tr_{2^n}(v)}: v\in \gf_{2^n}^{*} \right\}.$$ Moreover, for any $\gamma\in\Gamma_F$, $\gamma\equiv0\pmod4$ when $n\ge3$. Particularly, when $n=2k$ is even,  $\mathtt{DLU}_F = 2^{k}$.
\end{Th}

\begin{proof} By Proposition \ref{prop_monomials} it suffices to compute $\Delta_F(1,v)$. According to the definition,
	\begin{eqnarray*}
	   \Delta_F(1,v) & = & \sum_{x\in\gf_{2^n}}(-1)^{\tr_{2^n}(v(F(x)+F(x+1)))} \\
	   & = & \sum_{x\in\gf_{2^n}\backslash\{0,1\}} (-1)^{\tr_{2^n}\left(\frac{v}{x^2+x}\right)} + 2 \cdot (-1)^{\tr_{2^n}(v)}. 
	\end{eqnarray*}
Let sets $S_0$ and $S_1$ be $$S_0 := \left\{ \alpha\in\gf_{2^n}^{*} | \tr_{2^n}(\alpha) = 0 ~\text{and}~ \tr_{2^n}\left(\frac{v}{\alpha}\right) = 0    \right\},$$
and $$S_1 := \left\{ \alpha\in\gf_{2^n}^{*} | \tr_{2^n}(\alpha) = 0 ~\text{and}~ \tr_{2^n}\left(\frac{v}{\alpha}\right) = 1    \right\}.$$
It is clear that $\tr_{2^n}\left( \frac{v}{x^2+x} \right) = 0$ if and only if $ x^2+x \in S_0 $  and $\tr_{2^n} \left( \frac{v}{x^2+x} \right) = 1$ if and only if $x^2+x \in S_1$. Moreover, given $\alpha\in S_0$ or $S_1$, there exist two $x=x_0, x_0+1\in\gf_{2^n}$ satisfy $x^2+x=\alpha$.
Hence, 
$$\Delta_F(1,v) = 2\# S_0 - 2\# S_1 + 2 \cdot (-1)^{\tr_{2^n}(v)}. $$
In the following, we compute $\# S_0$ and $\# S_1$. We have 
\begin{eqnarray*}
4\cdot \# S_0 & = & \sum_{\alpha\in\gf_{2^n}^{*}}\sum_{a,b\in\gf_2}(-1)^{a\tr_{2^n}(\alpha)+b\tr_{2^n}\left(\frac{v}{\alpha}\right)}\\
&=& 2^n -1 + \sum_{\alpha\in\gf_{2^n}^{*}}(-1)^{\tr_{2^n}(\alpha)} + \sum_{\alpha\in\gf_{2^n}^{*}}(-1)^{\tr_{2^n}\left(\frac{v}{\alpha}\right)} + \sum_{\alpha\in\gf_{2^n}^{*}}(-1)^{\tr_{2^n}\left( \alpha+\frac{v}{\alpha} \right)}\\
&=& 2^n+ K\left(v \right) - 3,
\end{eqnarray*}
and thus $\# S_0 =  2^{n-2} + \frac{K\left(v \right) -3}{4}$. Similarly, $\# S_1 = 2^{n-2} - \frac{K\left(v \right) -3}{4}-1$ .

Therefore, 
$$\Delta_F(1,v) = K\left(v \right) - 1 + 2 \cdot (-1)^{\tr_{2^n}(v)}.$$
Furthermore, for any $\gamma\in\Gamma_F$,  $\gamma\equiv0\pmod4$ can be obtained directly from Lemma \ref{KK}. 

Particularly, when $n=2k,$ from Lemma \ref{K_range}, the maximum  of $K(v)$ is $2^{k+1}-1$ and in the case, $\tr_{2^n}(v)=0$ according to Lemma \ref{KK}. Thus the maximum of $\Delta_{F}(1,v)$ is $2^{k}$. Similarly, the minimum of $\Delta_{F}(1,v)$ is $-2^{k}$. Hence $\mathtt{DLU}_F = 2^{k}$ in this case.
\end{proof}

\subsection{Quadratic Functions}

\begin{Th}
	\label{quad_func}
	Let $F(x) = \sum_{0\le i<j\le n-1}a_{ij}x^{2^i+2^j}\in\gf_{2^n}[x]$. Then $$ \Gamma_F \in \left\{ -2^{n-1}, 0, 2^{n-1} \right\} ~~\text{and}~~ \mathtt{DLU}_F = 2^{n-1}.  $$
\end{Th}

\begin{proof}
	For any  $u,  v\in\gf_{2^n}^{*},$
	\begin{eqnarray*}
		\Delta_F(u,v) &=& \sum_{x\in\gf_{2^n}}(-1)^{\tr_{2^n}\left( v(F(x)+F(x+u))\right)  }\\
		&=&\sum_{x\in\gf_{2^n}}(-1)^{\tr_{2^n}\left( v\left( \sum_{0\le i<j\le n-1}  a_{ij}\left(  u^{2^j}x^{2^i} +  u^{2^i} x^{2^j}+u^{2^i+2^j}  \right) \right)  \right)}\\
		&=&(-1)^{\tr_{2^n}\left( v \left(  \sum_{0\le i<j\le n-1}  a_{ij} u^{2^i+2^j} \right)  \right)} \sum_{x\in\gf_{2^n}} (-1)^{\tr_{2^n}  \left(   L(u,v) x\right)}, \\
	\end{eqnarray*}
where $L(u,v) =  \sum_{0\le i<j\le n-1}\left( a_{ij}^{2^{-i}}u^{2^{j-i}}v^{2^{-i}} + a_{ij}^{2^{-j}} u^{2^{i-j}} v^{2^{-j}} \right).$	When $ L(u,v) = 0,$ $\Delta_F(u,v) = \pm 2^n$; otherwise, $\Delta_F(u,v) = 0.$ Thus $ \Gamma_F \in \left\{ -2^{n-1}, 0, 2^{n-1} \right\}. $ Moreover, since $F$ is not bent, from Remark \ref{Bent} we obtain $\mathtt{DLU}_F\neq0$ and then $\mathtt{DLU}_F = 2^{n-1}$.
\end{proof}

\begin{Cor}
	Let $F(x) = x^{2^i+1}\in\gf_{2^n}[x]$. Assume $d=\gcd(i,n)$ and $n=d\cdot n^{'}$. Then
		\begin{equation*}
	\Gamma_F =	\left\{
	\begin{array}{lr}
	\{ 0, 2^{n-1} \},  &~ \text{if}~ n^{'}~ \text{is even}, \\
	\{ -2^{n-1}, 0 \}, &~ \text{if}~ n^{'} ~\text{is odd and}~ d=1,\\
	\{ -2^{n-1}, 0, 2^{n-1} \}, &~ \text{otherwise.}
	\end{array}
	\right.
	\end{equation*}   
\end{Cor}

\begin{proof}
	From the proof of Theorem \ref{quad_func}, it is clear that 
	$$\Delta_{F}(1,v) = (-1)^{\tr_{2^n}(v)} \sum_{x\in\gf_{2^n}} (-1)^{\tr_{2^n}(L(v)x)}, $$
	where $L(v) = v^{2^{-i}}+v.$ Thus $\ker L = \gf_{2^{\gcd(i,n)}} = \gf_{2^d}$. Furthermore, for any $v\in\gf_{2^d}$, $\tr_{2^n}(v)=n^{'}\tr_{2^d}(v)$. Therefore, 
	\begin{equation*}
 \Delta_F(1,v) =	\left\{
	\begin{array}{lr}
	0,  &~  \text{if}~ v \in\gf_{2}^n\backslash\gf_{2}^d, \\
	2^n\cdot (-1)^{n^{'}\tr_{2^d}(v)}, &~ \text{if}~ v\in\gf_{2}^d.
	\end{array}
	\right.
	\end{equation*}
   Equivalently,
	\begin{equation*}
\Gamma_F =	\left\{
\begin{array}{lr}
\{ 0, 2^{n-1} \},  &~ \text{if}~ n^{'}~ \text{is even}, \\
\{ -2^{n-1}, 0 \}, &~ \text{if}~ n^{'} ~\text{is odd and}~ d=1,\\
\{ -2^{n-1}, 0, 2^{n-1} \}, &~ \text{otherwise.}
\end{array}
\right.
\end{equation*}   
\end{proof}

\subsection{Cubic Functions}

\begin{Th}
	Let $F(x) = x^{q^2+q+1}\in\gf_{q^4}[x]$, where $q=2^k$. Then
	$$\Gamma_F \in \left\{  -q^3, 0, q^3  \right\} ~~\text{and}~~ \mathtt{DLU}_F = q^3. $$
\end{Th}

\begin{proof}
	 For any $v\in\gf_{q^4}^{*}$,
	 \begin{eqnarray*}
	 \Delta_F(1,v) &=& \sum_{x\in\gf_{q^4}}(-1)^{\tr_{q^4} \left( v(F(x)+F(x+1) ) \right) } \\
	 &=& \sum_{x\in\gf_{q^4}}(-1)^{\tr_{q^4} \left( v\left( x^{q^2+q}+x^{q^2+1}+x^{q+1}+x^{q^2}+x^{q}+x+1 \right)  \right) }\\
	 &=& (-1)^{\tr_{q^4}(v)} \sum_{x\in\gf_{q^4}}(-1)^{\tr_{q^4}\left( vx^{q^2+1}+\left( v^{q^3}+v \right)x^{q+1} +\left( v^{q^3}+v^{q^2}+v  \right) x   \right) }
	 \end{eqnarray*}
    Moreover, 
    \begin{eqnarray*}
   && \Delta_F(1,v)^2 \\
     &=&  \sum_{x,y\in\gf_{q^4}}(-1)^{\tr_{q^4}\left( vx^{q^2+1}+\left( v^{q^3}+v \right)x^{q+1} +\left( v^{q^3}+v^{q^2}+v  \right) x +vy^{q^2+1}+\left( v^{q^3}+v \right)y^{q+1} +\left( v^{q^3}+v^{q^2}+v  \right) y   \right) } \\
    &= &  \sum_{x,y\in\gf_{q^4}}(-1)^{\tr_{q^4}\left( v(x+y)^{q^2+1}+\left( v^{q^3}+v \right)(x+y)^{q+1} +\left( v^{q^3}+v^{q^2}+v  \right) (x+y) +vy^{q^2+1}+\left( v^{q^3}+v \right)y^{q+1} +\left( v^{q^3}+v^{q^2}+v  \right) y   \right) } \\
    &=& \sum_{x,y\in\gf_{q^4}}(-1)^{\tr_{q^4}\left( v\left( x^{q^2+1}+xy^{q^2}+x^{q^2}y \right) + \left( v^{q^3}+v \right) \left(x^{q+1}+xy^{q}+x^qy  \right)  + \left( v^{q^3}+v^{q^2}+v \right) x \right) } \\
    &=&\sum_{x\in\gf_{q^4}}(-1)^{  \tr_{q^4} \left(  vx^{q^2+1} +\left(v^{q^3}+v \right)x^{q+1}+\left( v^{q^3}+v^{q^2}+v \right) x  \right)   }\sum_{y\in\gf_{q^4}}(-1)^{\tr_{q^4}(L_v(x)y)},
    \end{eqnarray*}	 
where $L_v(x) = \left( v^{q^3}+v^{q^2}\right)x^{q^3}+\left( v^{q^2}+v \right)x^{q^2}+\left( v^{q^3}+v\right)x^q.  $ Let 
$\ker\left(L_v \right) := \left\{  x\in\gf_{q^4} | L_v(x) =0 \right\}.$ Then 
$$\Delta_F(1,v)^2  = q^4\cdot\sum_{x\in\ker\left(L_v \right)} (-1)^{ \phi_v(x) },  $$ where $\phi_v(x) = \tr_{q^4} \left(  vx^{q^2+1} +\left(v^{q^3}+v \right)x^{q+1}+\left( v^{q^3}+v^{q^2}+v \right) x  \right). $  

(1) When $v\in\gf_q^{*}$, $L_v(x)=0$ and thus $\ker\left(L_v \right) = \gf_{q^4}$. Moreover, $ \phi_v(x) = \tr_{q^4}\left(  vx^{q^2+1} +vx \right) = \tr_{q^4}\left(vx\right). $ Therefore, $$\Delta_F(1,v)^2  = q^4\cdot\sum_{x\in \gf_{q^4}} (-1)^{ \tr_{q^4}(vx) } = 0.$$

(2) When $v\in\gf_{q^4}\backslash\gf_{q}$, $\phi_v(x)$ is linear on $\ker\left(L_v \right)$, which can be proved by direct computations. Thus $\Delta_F(1,v)^2\neq 0$ only when $\phi_v(x)$ is zero mapping on $\ker\left(L_v \right)$. In addition, according to Remark \ref{Bent}, there must exist some $v$ such that $\Delta_F(1,v)\neq0$ since $F(x)$ is not bent. Moreover, the Dickson Matrix of $L_v$ is 
$$D = \begin{pmatrix}
0 & v^{q^3}+v &  v^{q^2}+v & v^{q^3}+v^{q^2}\\
v^{q^3}+v & 0 & v^q+v & v^{q^3}+v^q \\
v^{q^2}+v & v^q+v & 0 & v^{q^2}+v^q \\
v^{q^3}+v^{q^2} & v^{q^3}+v^q & v^{q^2}+v^q & 0 
\end{pmatrix}$$
It is easy to compute that the Rank of $D$ is $2$ and thus $\#\ker\left(L_v\right)=q^2$. Therefore, there exists some $v$ with $$\Delta_F(1,v)^2 = q^4\cdot\sum_{x\in\ker\left(L_v \right)} (-1)^{ \phi_v(x)} =q^4\cdot \#\ker\left(L_v\right) = q^6. $$

Thus $\Gamma_F \in \left\{  -q^3, 0, q^3  \right\}$ and $\mathtt{DLU}_F = q^3.$
\end{proof}

\section{Conclusion}
\label{conclusion}
{This paper intensively investigate the differential-linear connectivity table (DLCT) of vectorial Boolean functions.
The main contributions of this paper are fourfold.
Firstly we establish the connection of DLCT with the (generalized) additive autocorrelation of vectorial Boolean functions 
and characterize the relation between DLCT and the Walsh transform and differential distribution table.
Secondly we give some generic bounds about the differential-linear uniformity of vectorial Boolean functions,
study the invariance property of DLCT under the affine, EA and CCZ equivalence and exhaust the DLCT of optimal $4\times 4$ S-boxes. 
Thirdly, we further characterize properties of the DLCT of cryptographically desirable funcitons, including monomials, APN, plateaued and AB functions and convert the DLCT of APN and AB functions to the Walsh transform of certain Boolean functions. Finally, we investigate the DLCT spectra of some polynomials over $\gf_{2^n}$ in certain forms, such as the inverse, Gold and Bracken-Leander functions. 

With the importance of the new DLCT criterion of vectorial Boolean functions, our present work only covers a rather portion of interesting problems in this direction. We think that many problems deserve further research. For instance,  
the generic bounds of DLU of vectorial Boolean functions are significantly less than experimental results and should be further improved. A natural follow-up topic would be the investigation and construction of optimal, or near-optimal, vectorial Boolean functions with respect to the bounds.
}


\begin{thebibliography}{(1)}
	
\bibitem{Bud2014} L. Budaghyan, \newblock Construction and analysis of cryptographic functions, \newblock {\em New York, NY, USA: Springer-Verlag}, 2014.

	\bibitem{BC2018} C. Boura, A. Canteaut, \newblock On the boomerang uniformity of cryptographic sboxes, \newblock {\em  IACR Trans. Symmetric Cryptol.}, 3 (2018), pp. 290-310.

\bibitem{BCP2006} L. Budaghyan, C. Carlet and A. Pott, \newblock New classes of almost bent and almost perfect nonlinear polynomials, \newblock {\em  	IEEE Trans. Inf. Theory}, 52 (2006), pp. 1141-1152.

\bibitem{DLCT2019} A. Bar-On, O. Dunkelman, N. Keller and et al., \newblock DLCT: A new tool for differential-linear cryptanalysis,  \newblock {\em in: Advances in Cryptology - EUROCRYPT 2019}, in: LNCS. vol. 11476, 2019, pp. 313-342.

\bibitem{Biham1991} E. Biham, A. Shamir, \newblock Differential cryptanalysis of DES-like cryptosystems, \newblock{ \em J. Cryptology}, 4 (1991), pp. 3-72.


\bibitem{Carlet2010} C. Carlet, \newblock Boolean functions for cryptography and error correcting
codes, \newblock in Boolean Models and Methods in Mathematics, Computer Science, and Engineering, Y. Crama and P. L. Hammer, Eds. Cambridge, U.K.: Cambridge Univ. Press, 2010, pp. 257–397. [Online]. Available: http://www.math.univ-paris13.fr/~carlet/pubs.html.

\bibitem{Carlet2015} C. Carlet, \newblock Boolean and Vectorial plateaued functions and APN functions, \newblock {\em  	IEEE Trans. Inf. Theory}, 61 (11) (2015), pp. 6272-6289.

\bibitem{BCT2018} C. Cid, T. Huang, T. Peyrin and et al., \newblock Boomerang Connectivity Table: A New Cryptanalysis Tool, \newblock {\em in: Advances in Cryptology - EUROCRYPT 2018}, in: LNCS. vol. 10821, 2018, pp. 683-714.

\bibitem{CHZ2007} P. Charpin, T. Helleseth and V. Zinoviev, \newblock Propagation characteristics of $x^{-1}\to x$ and Kloosterman sums, \newblock {\em  Finite Fields Appl.}, 13 (2007), pp. 366-381.

\bibitem{CV1995} F. Chabaud, S.Vaudenay, \newblock  Links between differential and linear cryptanalysis, \newblock  {\em in: Advances in Cryptology -EUROCRYPT'94}, in: LNCS, Springer-Verlag, New York, vol. 950, 1995, pp. 356–365.

\bibitem{Dillon} J. Dillon, \newblock Multiplicative Difference Sets via Characters, {\em Designs, Codes and
Cryptography}, 17(1999): 225-235.

\bibitem{GK2004} G. Guang, K. Khoongming, \newblock  Additive autocorrelation of resilient Boolean functions, \newblock  {\em In: Selected Areas in Cryptography 2003}, in: LNCS, Springer-Verlag, Berlin, vol. 3006, 2004, pp. 275–290.
 
\bibitem{HZ1999} T. Helleseth, V. Zinoviev, \newblock On $Z_4$-linear goethals codes and Kloosterman sums, \newblock {\em Des. Codes Cryptogr.}, 17 (1999), pp. 269-288.

\bibitem{L2008} P. Lison\v{e}k, \newblock On the connection between Kloosterman sums and elliptic curves, \newblock {\em SETA 2008}, in: LNCS. vol. 5203, 2008, pp. 182-187.

\bibitem{LH1994}  S. K. Langford, M. E. Hellman, \newblock Differential-Linear Cryptanalysis, \newblock {\em in: Advances in Cryptology - CRYPTO 1994}, in: LNCS. vol. 839, 1994, pp. 17-25.


\bibitem{LP2007} G. Leander, A. Poschmann, \newblock On the Classification of 4 Bit S-Boxes, \newblock {\em  In: WAIFI 2007}, in: LNCS. vol. 4547, 2007, pp. 159-176. 


\bibitem{LQSL2019} K. Li, L. Qu, B. Sun and et al., \newblock New Results about the Boomerang Uniformity of Permutation Polynomials,  \newblock  {\em IEEE Trans. Inf. Theory,} 2019, doi:  10.1109/TIT.2019.2918531. 

\bibitem{LW1990} G. Lachaud, J. Wolfmann, \newblock The weights of the orthogonals of the extended quadratic binary Goppa codes, \newblock {\em IEEE Trans. Inf. Theory}, 36 (3) (1990), pp. 686-692.

\bibitem{Sihem2016} S. Mesnager, \newblock Bent functions: fundamentals and results. \newblock {\em Springer, Switzerland,} 2016.

\bibitem{MTX2019} S. Mesnager, C. Tang and M. Xiong, \newblock On the boomerang uniformity of (quadratic) permutations over $\gf_{2^n}$, \newblock  {\em arXiv: 1903. 00501v1,} 2019. 
	
\bibitem{Nyberg1995} K. Nyberg, \newblock S-Boxes and round functions with controllable linearity and differential uniformity, \newblock {\em in: Fast Software Encryption-FSE 1994}, in LNCS. vol. 1008, Springer-Verlag, Berlin, Germany, 1995, pp. 111-130.

\bibitem{Rothaus1976} O. S. Rothaus, \newblock On 'bent' functions, \newblock {J. Combinat. Theory A,} 3 (1976), pp. 300-305.

\bibitem{SQH2019} L. Song, X. Qi and L. Hu,  \newblock Boomerang Connectivity Table Revisited-Application to SKINNY and AES, \newblock  {\em  IACR Trans. Symmetric Cryptol.}, 1 (2019), pp. 118 - 141.

\bibitem{Trach1970} H. M. Trachtenberg, \newblock On the Cross-Correlation Functions of Maximal Linear Sequences, Ph.D. dissertation, University of Southern California, Los Angeles, 1970.


\bibitem{ZZ1995} X. M. Zhang, Y. Zheng, \newblock GAC --- the criterion for global avalanche characteristics and nonlinearity of cryptographic functions, \newblock {\em Journal of Universal Computer Science}, 1 (1995), pp. 136 - 150.
\end{thebibliography}
\end{document}